\documentclass[12pt]{article}
\usepackage{amssymb}
\usepackage{amsmath}
\usepackage{amsthm}
\usepackage{calc}
\usepackage{xcolor}
\usepackage{cancel}
\usepackage{tikz,pgfplots}
\usepackage{tkz-euclide}
\usetikzlibrary{positioning,shapes,fit,arrows}
\usetikzlibrary {arrows.meta}

\usepackage[affil-it]{authblk}
 
\usepackage{graphicx}

\newtheorem{theorem}{Theorem}
\newtheorem{prop}[theorem]{Proposition}
\newtheorem{lem}[theorem]{Lemma}
\newtheorem{cor}[theorem]{Corollary}

\newtheorem{remark}[theorem]{Remark}
\newenvironment{rem}{\begin{remark} \rm}{\end{remark}}
\newtheorem{example}[theorem]{Example}

\def\d{{\mathrm d}}

\usepackage{cancel}

\def\parpo#1#2{\{#1,#2\}}

\newcommand{\RR}{{{\mathbb{R}}}}
\newcommand{\brke}[2]{{{\langle #1,#2\rangle}}}

\def\M{{\mathcal M}}
\def\Tr{\operatorname{Tr}}
\def\cycsum{\sum_{\circlearrowleft(\alpha,\beta,\gamma)}}

\begin{document}
\title{
An involutivity theorem for a class\\ of Poisson quasi-Nijenhuis manifolds
}
\date{} 

\author{E.\ Chu\~no Vizarreta${}^{1}$, G.\ Falqui${}^{2,5}$, I.\ Mencattini${}^3$, 
M.\ Pedroni${}^{4,5}$
}

\affil{
{\small  $^1$Unidade Acad\^emica de Belo Jardim, Universidade Federal Rural de Pernambuco, Brazil}\\
{\small eber.vizarreta@ufrpe.br  
}\\
\medskip
{\small  $^2$Dipartimento di Matematica e Applicazioni, Universit\`a di Milano-Bicocca, Italy
}\\
{\small  gregorio.falqui@unimib.it 
}\\
\medskip
{\small $^3$Instituto de Ci\^encias Matem\'aticas e de Computa\c c\~ao,  Universidade de S\~ao Paulo, Brazil}\\
{\small igorre@icmc.usp.br 
}\\
\medskip
{\small $^4$Dipartimento di Ingegneria Gestionale, dell'Informazione e della Produzione,  Universit\`a di Bergamo, Italy}\\
{\small marco.pedroni@unibg.it 
}\\
\medskip
{\small  $^5$INFN, Sezione di Milano-Bicocca, 
Milano, Italy}
}

\maketitle
\abstract{This 
note aims to continue our study about the applications of Poisson quasi-Nijenhuis geometry to the theory of classical completely integrable systems. More precisely, we will present new versions of the deformation and involutivity theorems, under the hypothesis that the closed 2-form triggering the deformation and the closed 3-form defining the Poisson quasi-Nijenhuis structure are factorized. These results will be supplemented by several 
examples of involutive Poisson quasi-Nijenhuis manifolds.
\noindent
}


\baselineskip=0,6cm

\section{Introduction}

The notion of a \emph{Poisson-Nijenhuis}  structure, introduced in \cite{MagriMorosiRagnisco85}, see also \cite{KM}, is a geometrical framework to describe complete integrability of classical Hamiltonian systems. In a few words, the reason for this state of affairs can be summarized as follows: the presence of a recursion operator, i.e., a \emph{torsion free} (1,1) tensor in the definition of a Poisson-Nijenhuis structure, guarantees the existence of a distinguished set of functions which are in involution with respect to all the Poisson brackets associated to the underlying structure. 
It is a natural and interesting problem to look for geometric structures weaker than the Poisson-Nijenhuis ones but that can still be used to picture classical complete integrability.

An interesting proposal, aimed to generalize the notion of a Poisson-Nijenhuis structure, was done 
in \cite{SX}, where \emph{Poisson quasi-Nijenhuis} structures were introduced.
For 
these structures, the torsion of the corresponding (1,1) tensor is not zero but it is controlled by the presence of a closed 3-form, see next section for a precise definition. However, this is not enough to ensure 
the existence of a distinguished set of Poisson commuting functions, characteristic of the Poisson-Nijenhuis case. 

Rephrasing what we wrote above, 
and echoing Remark 3.13 in \cite{SX}, one could say that it is an interesting problem to pin down a set of conditions warranting the existence of a (maximal) set of Poisson commutative functions on a Poisson quasi-Nijenhuis manifold, to link its geometry to the theory of classical integrable systems.

In a series of papers \cite{FMOP2020,FMP2023, DMP2024, CFMP2025} we started to investigate this problem. Inspired by the example of the classical Toda systems, we introduced the notion of an involutive Poisson quasi-Nijenhuis structure. Moreover, we found a set of sufficient conditions for a Poisson quasi-Nijenhuis structure to be involutive, see Theorem \ref{thm:invthmv}, and we 
described a way to deform a Poisson quasi-Nijenhuis structure into another one, see Theorem \ref{thm:defthmv}.


The present note is meant to continue this line of investigation and 
has the following two goals:\\ 
1) to study the deformation and involutivity problems under suitable factorization assumptions; \\
2) to provide new examples of involutive Poisson quasi-Nijenhuis structures.  

In more details, the plan of the paper is as follows.  In Section 2 we will recall the basic notions of Poisson-Nijenhuis and Poisson quasi-Nijenhuis geometry and we will briefly comment on how Poisson-Nijenhuis structures enter in the theory of classical completely integrable systems. In Section 3 we will 
prove a version of the 
deformation theorem 
(i.e., Theorem \ref{thm:defthmv}), under suitable factorization hypothesis on the closed 2-form triggering the deformation, see Proposition \ref{prop:ab_def}. An application of the latter can be found in Examples \ref{exa:tsiganov} and \ref{exa:ourToda}, both devoted to the Toda system (open and closed). 
Section 4 encloses the most important result of this note, which consists of a new version of the 
involutivity theorem 
(i.e., Theorem \ref{thm:invthmv}). This is proved, again, under a suitable factorization hypothesis on the 
closed 3-form defining the Poisson quasi-Nijenhuis structure, see Theorem \ref{thm:inv}. This result 
is applied 
in Examples \ref{exa:inv1}, \ref{exa:inv2}, \ref{exa:inv3} and \ref{exa:inv4}. 
We would like to stress that the 
Poisson quasi-Nijenhuis structures 
described in Examples \ref{exa:inv2} and \ref{exa:inv4} are, to the best of our knowledge, new. In particular, the one explored in Example \ref{exa:inv4}, 
obtained by deforming the Poisson-Nijenhuis structure 
of the open Toda system of type $D_n$, 
describes an integrable system that does not seem to be associated with affine Lie algebras (see \cite{RSTS}, Table 3). 
We 
also emphasize that the 2-form used to construct the involutive Poisson quasi-Nijenhuis structure in Example \ref{exa:inv4}
is the same 2-form used in the deformation of the Poisson-Nijenhuis structure describing the open Toda system of type $C_n$, see \cite{CFMP2025}.
Finally, Example \ref{exa:inv5} provides an application of Corollary \ref{cor:inv_def} (a consequence of Theorem \ref{thm:inv}) and 
was inspired by Section 6 of \cite{Crampin-Sarlet-Thompson}.


\section{Preliminaries on PN structures, PqN structures and their deformations}
\label{sec:PN-PqN}
In this section we recall some definitions and results about Poisson-Nijenhuis and Poisson quasi-Nijenhuis structures, together with a deformation procedure driven by a closed 2-form. Moreover, we 
comment on the use of the Poisson-Nijenhuis structures in the theory of classical integrable systems. Hereafter, by Poisson-Nijenhuis (bi-Hamiltonian, Poisson quasi-Nijenhuis) manifold we mean a manifold endowed with a Poisson-Nijenhuis (bi-Hamiltonian, Poisson quasi-Nijenhuis) structure.

If $N:T\M\to T\M$ is a $(1,1)$ tensor field on a manifold $\M$, then its {\it Nijenhuis torsion\/} is defined as 
\begin{equation}
\label{tndef1}
T_N(X,Y)=[NX,NY]-N\left([NX,Y]+[X,NY]-N[X,Y]\right).
\end{equation}
Given a $p$-form $\alpha$, with $p\ge 1$, one can construct another $p$-form $i_N\alpha$ as 
\begin{equation}
\label{iNalpha}
i_N\alpha(X_1,\dots,X_p)=\sum_{i=1}^p \alpha(X_1,\dots,NX_i,\dots,X_p).
\end{equation}
If $i_N f=0$ for all functions $f$, then $i_N$ is a derivation of degree zero.
Moreover, the usual Cartan differential can be modified as 
\begin{equation}
\label{eq:dNd}
\d_N=i_N\circ \d-\d\circ i_N.
\end{equation}
It turns out that $
\d_N
^2 =0$ if and only if the torsion of $N$ vanishes.

Suppose now that $(\M,\pi)$ is a Poisson manifold. Then one can define a Lie bracket between 1-forms on $\M$ as
\begin{equation}
\label{eq:liealgpi}
[\alpha,\beta]_\pi=L_{\pi^\sharp\alpha}\beta-L_{\pi^\sharp\beta}\alpha-\d\langle\beta,\pi^\sharp\alpha\rangle,
\end{equation}
where $\pi^\sharp:T^*\M\to T\M$ is defined by $\langle \beta,\pi^\sharp\alpha\rangle=\pi(\alpha,\beta)$.
This Lie bracket can be uniquely extended to all forms on $\M$ in such a way that, if $\eta$ is a $q$-form and $\eta'$ is a $q'$-form, then 
$[\eta,\eta']_\pi$ is a $(q+q'-1)$-form and the following conditions are satisfied:
\begin{itemize}
\item[(K1)] $[\eta,\eta']_\pi=-(-1)^{(q-1)(q'-1)}[\eta',\eta]_\pi$; 
\item[(K2)] $[\alpha,f]_\pi=i_{\pi^\sharp\alpha}\,\d f=\langle \d f,\pi^\sharp\alpha\rangle$ for all $f\in C^\infty(\M)$ and for all 1-forms $\alpha$;
\item[(K3)] 
$[\eta,\cdot]_\pi$ 
is a derivation of degree $q-1$ of the wedge product, that is, 
for any differential form $\eta''$,
\begin{equation}
\label{deriv-koszul}
[\eta,\eta'\wedge\eta'']_\pi=[\eta,\eta']_\pi\wedge\eta''+(-1)^{(q-1)q'}\eta'\wedge[\eta,\eta'']_\pi.
\end{equation}
\end{itemize}
This extension is a {\it graded\/} Lie bracket, in the sense that (besides (K1)) the graded Jacobi identity holds:
\begin{equation}
\label{graded-jacobi}
(-1)^{(q_1-1)(q_3-1)}[\eta_1,[\eta_2,\eta_3]_\pi]_\pi+(-1)^{(q_2-1)(q_1-1)}[\eta_2,[\eta_3,\eta_1]_\pi]_\pi+(-1)^{(q_3-1)(q_2-1)}[\eta_3,[\eta_1,\eta_2]_\pi]_\pi=0,
\end{equation}
where $q_i$ is the degree of $\eta_i$.
It is sometimes called the Koszul bracket --- see, e.g., \cite{FiorenzaManetti2012} and references therein. Its relation with the Poisson bracket is
$$
[\d f,\d g]_\pi=\d\{f,g\},
$$
where 
\begin{equation}
\{f,g\}=\pi(\d f,\d g)=\langle \d g,\pi^\sharp\d f\rangle.\label{eq:poi}
\end{equation}

A Poisson bivector $\pi$ and a $(1,1)$ tensor field $N$ on a manifold $\M$ are said to be {\it compatible\/}, see \cite{MagriMorosiRagnisco85,KM}, if 
\begin{equation}
\label{N-P-compatible}
\begin{split}
&N\pi^\sharp=\pi^\sharp N^*\,,\qquad
\mbox{where $N^*:T^*\M\to T^*\M$ is the transpose of $N$;}\\
&L_{\pi^\sharp\alpha}(N) X-\pi^\sharp
L_{X}(N^*\alpha)+\pi^\sharp
L_{NX}\alpha=0,\qquad\mbox{for all 1-forms $\alpha$ and vector fields $X$.}
\end{split}
\end{equation}
Notice that the first condition means that $N\pi^\sharp$ is skewsymmetric, and entails that $N^l\pi^\sharp$ is skewsymmetric too, for all $l\ge 0$. 
It was proved in \cite{YKS96} that conditions \eqref{N-P-compatible} hold if and only if $\d_N$ is a derivation of $[\cdot,\cdot]_\pi$, 
that is,
\begin{equation}
\label{deriv-wedge}
\d_N[\eta,\eta']_\pi=[\d_N\eta,\eta']_\pi+(-1)^{(q-1)}[\eta,\d_N\eta']_\pi
\end{equation}
if $\eta$ is a $q$-form and $\eta'$ is any differential form. 

A \emph{Poisson-Nijenuis} manifold (PN, from now on) 
is a triple $(\M,\pi,N)$, where $\M$ is a smooth manifold, $\pi$ is a Poisson tensor compatible with $N$, see \eqref{N-P-compatible}, and the torsion of $N$ vanishes, see \eqref{tndef1}. Every PN manifold is endowed with a sequence $\{\pi_k\}_{k\geq 1}$ of pairwise compatible Poisson bivector fields, where $\pi_k^\sharp=N^k\pi^\sharp$ for all $k\geq 1$, and 
the notion of compatibility between two bivector fields is expressed by the vanishing of their Schouten bracket, see for example \cite{DZ}. Note that every PN manifold is, in particular, bi-Hamiltonian, i.e., endowed with a pair of compatible Poisson bivector fields, say $\pi_0=\pi$ and $\pi_1$. On such a manifold, if $\{H_k\}_{k\geq 1}$ is a sequence of functions satisfying 
\begin{equation}
	\pi_1^\sharp\d H_k=\pi_0^\sharp\d H_{k+1}, \quad\forall k\geq 1,\label{eq:lmc}
\end{equation}
then, for all $l,r\geq 1$,
\[
\{H_l,H_r\}_1=0=\{H_l,H_r\}_0,
\]
where the brackets $\{\cdot,\cdot\}_i$, $i=0,1$, are defined in \eqref{eq:poi}. A sequence of functions satisfying \eqref{eq:lmc} is called a \emph{Lenard-Magri chain} relative to $(\pi_0,\pi_1)$. 

A PN manifold, being endowed with an enhanced bi-Hamiltonian structure, has a distinguished Lenard-Magri chain $\{H_k\}_{k\geq 1}$, where $H_k=\frac1{2k}{\Tr N^k}$. In fact, since $T_N= 0$, one can prove that 
\begin{equation}
	N^\ast\d H_k=\d H_{k+1}, \quad\forall k\geq 1,\label{em:pnlm}
\end{equation}
which immediately imply \eqref{eq:lmc}. 
In summary, on a PN manifold $(\M,\pi,N)$ one can find a rich supply of Poisson commuting functions, 
and this is the reason of the importance of PN structures in the theory of classical integrable systems. 


To generalize the PN setting, Sti\'enon and Xu \cite{SX} defined a quadruple $(\M,\pi,N,\phi)$ to be a {\it Poisson quasi-Nijenhuis 
manifold\/} (PqN, from now on) 
if:
\begin{itemize}
\item the Poisson bivector $\pi$ and the $(1,1)$ tensor field $N$ are compatible;
\item the 3-forms $\phi$ and $i_N\phi$ are closed;
\item $T_N(X,Y)=\pi^\sharp\left(i_{X\wedge Y}\phi\right)$ for all vector fields $X$ and $Y$, where $i_{X\wedge Y}\phi$ is the 1-form defined as $\langle i_{X\wedge Y}\phi,Z\rangle=\phi(X,Y,Z)$.
\end{itemize}

\begin{rem}
\begin{itemize}	
\item[] 
\item[(i)] 
A consequence of $N$  having non-zero torsion 
is that $\pi$ is the only Poisson tensor associated to a PqN structure. Unless differently stated, hereafter we will denote with $\{\cdot,\cdot\}$ the corresponding Poisson bracket.
\item[(ii)] A slightly more general definition of PqN manifold was recently proposed in \cite{BursztynDrummond2019}.
\end{itemize}
\end{rem}

Since $T_N\neq 0$ on a PqN manifold, the relations  
in \eqref{em:pnlm} do not hold anymore but 
are replaced by 
\begin{equation}
N^\ast\d H_k=\d H_{k+1}+\phi_{k-1},\qquad k\geq 1,\label{eq:pqngen}
\end{equation}
where the $H_k$ are defined as above and the 1-forms $\phi_k$ are defined by 
\begin{equation}
\langle\phi_k,X\rangle=\frac{1}{2}\Tr(N^ki_X T_N),\qquad k\geq 0,\label{eq:phi_k}
\end{equation}
where  $i_XT_N(Y)=T_N(X,Y)$, for $X,Y$ vector fields on $\M$. Because of the presence of the 1- forms $\phi_k$, 
relations \eqref{eq:pqngen} do not yield, in general, the Poisson involutivity of the functions $H_k$ and for this reason a \emph{general} PqN structure is not suitable to describe complete integrability. As mentioned in the Introduction, inspired
by the case of the Toda systems of type $A_n$, 
we defined an \emph{involutive} PqN structure as one whose functions $\{H_k\}_{k\geq 1}$ form a Poisson commutative family with respect to the unique Poisson structure available, and we proved (see \cite{CFMP2025}) the following 
\begin{theorem}[Involutivity theorem]\label{thm:invthmv}
	Let $(\M,\pi,N,\phi)$ be a PqN manifold and let $H_k=\frac{1}{2k}\Tr N^k$. Suppose that there exists a 2-form $\Omega$ such that:
	\begin{enumerate}
		\item[(a)] $\phi=-2\d H_1\wedge\Omega$;
		\item[(b)] $\Omega(X_j,Y_k)=0$ for all $j,k\geq 1$, where $Y_k=N^{k-1}X_1-X_k$ and $X_k=\pi^\sharp \d H_k$.
	\end{enumerate}
	Then $\{H_j,H_k\}=0$ for all $j,k$.
\end{theorem} 


We conclude this section by recalling a result about the deformation of PqN manifolds. 
It was proved in \cite{DMP2024}, generalizing that in \cite{FMP2023}, where the starting point is a PN manifold. 
For its applications to 
Toda lattices, see \cite{FMP2023,CFMP2025} and Section \ref{sec:inv_thm}.
 
\begin{theorem}[Deformation theorem]\label{thm:defthmv}
\label{thm:gim}
Let $(\M,\pi,N,\phi)$ be a PqN manifold and let $\Omega$ be a closed 2-form. Define as usual $\Omega^\flat:T\M\to T^*\M$ as 
$\Omega^\flat(X)=i_X\Omega$. If 
\begin{equation}
\label{N-phi-hat}
\widehat N=N+\pi^\sharp\,\Omega^\flat\qquad\mbox{and}\qquad\widehat\phi=\phi+\d_N\Omega+\frac{1}{2}[\Omega,\Omega]_\pi,
\end{equation} 
then $(\M,\pi,\widehat N,\widehat\phi)$ is a PqN manifold. 
\end{theorem}


\section{Deformations with factorized 2-forms}
\label{sec:def-factorized}


In this section we consider some special deformations of a PqN manifold, governed by 2-forms 
which are the wedge product of two 1-forms. We start with

\begin{lem}
\label{lem:f-g-I}
Let $\pi$ be a Poisson tensor on $\M$ and let $N$ be a (1,1) tensor field compatible with $\pi$. 
Suppose that $\alpha$, $\beta$, and $\gamma$ are 1-forms on $\M$ such that
$$
\d_N\alpha=\alpha\wedge \gamma,\qquad \d_N\beta=\beta\wedge \gamma
.
$$
If
$$
\delta=\beta+\frac12[\beta,\alpha]_\pi\quad\mbox{and}\quad\Omega=\alpha\wedge\delta,
$$
then 
\begin{equation}
\label{eq:Maurer-Cartan}
\d_N\Omega+\frac12[\Omega,\Omega]_\pi=\alpha\wedge\epsilon,
\end{equation}
where
\begin{equation}
\label{eq:epsilon}
\epsilon=2\gamma\wedge \beta+\frac32\gamma\wedge [\beta,\alpha]_\pi+\frac12 \beta\wedge [\alpha,\gamma]_\pi+
[\delta,\alpha]_\pi\wedge\delta.
\end{equation}
\end{lem}
\begin{proof}
Using \eqref{deriv-wedge} and the properties of the Koszul bracket, we have that 
\begin{equation*}
\begin{split}
\d_N\Omega
&=\d_N(\alpha\wedge\delta)
=\d_N\alpha\wedge\delta-\alpha\wedge\d_N\delta
=\alpha\wedge\gamma\wedge\left(\beta+\frac12[\beta,\alpha]_\pi\right)-\alpha\wedge\left(\d_N\beta+\frac12\d_N[\beta,\alpha]_\pi\right)\\
&=\alpha\wedge\left\{\gamma\wedge\left(\beta+\frac12[\beta,\alpha]_\pi\right)-\beta\wedge\gamma
    -\frac12[\d_N\beta,\alpha]_\pi-\frac12[\beta,\d_N\alpha]_\pi\right\}\\
&=\alpha\wedge\left(2\gamma\wedge\beta+\frac12\gamma\wedge[\beta,\alpha]_\pi-\frac12[\beta\wedge\gamma,\alpha]_\pi
    -\frac12[\beta,\alpha\wedge\gamma]_\pi\right)\\
&=\alpha\wedge\left(2\gamma\wedge\beta+\frac12\gamma\wedge[\beta,\alpha]_\pi+\frac12[\alpha,\beta\wedge\gamma]_\pi
    -\frac12[\beta,\alpha\wedge\gamma]_\pi\right)\\
&=\alpha\wedge\left(2\gamma\wedge\beta+\frac12\gamma\wedge[\beta,\alpha]_\pi+\frac12[\alpha,\beta]_\pi\wedge\gamma
    +\frac12\beta\wedge[\alpha,\gamma]_\pi-\frac12[\beta,\alpha]_\pi\wedge\gamma-\frac12\alpha\wedge[\beta,\gamma]_\pi\right)\\
&=\alpha\wedge\left(2\gamma\wedge\beta+\frac32\gamma\wedge[\beta,\alpha]_\pi+\frac12\beta\wedge[\alpha,\gamma]_\pi\right).
\end{split}
\end{equation*}
With similar computations one obtains
\begin{equation*}
[\Omega,\Omega]_\pi=[\alpha\wedge\delta,\alpha\wedge\delta]_\pi=2\alpha\wedge[\delta,\alpha]_\pi\wedge\delta.
\end{equation*}
Hence we have that $\d_N\Omega+\frac12[\Omega,\Omega]_\pi=\alpha\wedge\epsilon$, where $\epsilon$ is given by \eqref{eq:epsilon}.
\end{proof}

\begin{prop}
\label{prop:ab_def}
Let $(\M,\pi,N,\phi)$ be a PqN manifold. Suppose that $\alpha$ and $\beta$ are closed 1-forms on $\M$ such that
\begin{equation}
\label{recurrence}
\d_N\alpha=\alpha\wedge \beta,\qquad \d_N\beta=0.
\end{equation}
If
\begin{equation}
\label{eq:deltaOmega}
\delta=\beta+\frac12[\beta,\alpha]_\pi\quad\mbox{and}\quad\Omega=\alpha\wedge\delta,
\end{equation}
then $(\M,\pi,\widehat N,\widehat \phi)$ is a PqN manifold, where
$$
\widehat N=N+\pi^\sharp\Omega^\flat,\qquad
\widehat \phi=
\phi+\frac12\alpha\wedge \left[[\beta,\alpha]_\pi,\alpha\right]_\pi\wedge\delta=
\phi+\frac12\left[[\alpha,\beta]_\pi,\alpha\right]_\pi\wedge\Omega.
$$
\end{prop}
\begin{proof}
It is sufficient to apply Lemma \ref{lem:f-g-I} in the particular case $\gamma=\beta$ and then Theorem \ref{thm:gim}. Notice that $\Omega$ is closed since $\alpha$ and $\delta$ are closed. The 3-form of the deformed PqN manifold is 
$$
\widehat \phi=\phi+\d_N\Omega+\frac12[\Omega,\Omega]_\pi=\phi+\alpha\wedge\epsilon,
$$
and it turns out that 
$$
\epsilon= \frac12\left[[\beta,\alpha]_\pi,\alpha\right]_\pi\wedge\delta.
$$
\end{proof}


\begin{rem}
\label{rem:PNtoPN}
As an easy consequence of the previous result, if $(\M,\pi,N)$ is a PN manifold and $\left[[\beta,\alpha]_\pi,\alpha\right]_\pi=0$, then the deformed structure, by means of 
\begin{equation}
\label{eq:Omega-PNtoPN}
\Omega=\alpha\wedge\left(\beta+\frac12[\beta,\alpha]_\pi\right),
\end{equation}
is also PN.
\end{rem}

The next two examples are devoted to PN and PqN structures related to open and closed Toda lattices.

\begin{example}
\label{exa:tsiganov}
Suppose that  $({\RR}^{2n},\pi,N
)$ is the Das-Okubo PN structure \cite{DO} of the open Toda lattice. This means that $\pi$ is the canonical Poisson tensor, 
i.e., $\{p_i,q_j\}=\delta_{ij}$, and 
\begin{equation}
\label{N-closedToda}
\begin{split}
N
&=\sum_{i=1}^{n}p_i\left(\partial_{q_i}\otimes \d q_i+\partial_{p_i}\otimes \d p_i\right)+
\sum_{i<j}\left(\partial_{q_i}\otimes \d p_j-\partial_{q_j}\otimes \d p_i\right)\\
&+\sum_{i=1}^{n-1}\,e^{q_i-q_{i+1}}\left(\partial_{p_{i+1}}\otimes \d q_i-\partial_{p_i}\otimes \d q_{i+1}\right).
\end{split}
\end{equation}
If $f=e^{q_n}-e^{-q_1}$ and $g=\sum_{i=1}^n p_i$, then one can check that
$$
\d_N\d f=\d f\wedge\d g,\qquad \d_N\d g=0,\qquad \{\{g,f\},f\}=0.
$$
The second assertion follows from $\d_N H_1=N^*\d H_1=\d H_2$, where $H_k=\frac{1}{2k}\Tr (N^k)$; 
the third one from $\{g,f\}=e^{q_n}+e^{-q_1}$. Thus 
Remark \ref{rem:PNtoPN} can be applied, with $\alpha=\d f$ and $\beta=\d g$. 
In this case, the 2-form to be used to deform the Das-Okubo PN structure into another PN structure,
see \eqref{eq:Omega-PNtoPN}, is 
$\Omega_{\mbox{\scriptsize{\rm Ts}}}=\d f\wedge\d h$, where 
$$
h=g+\frac12\{g,f\}=\sum_{i=1}^n p_i+\frac12(e^{q_n}+e^{-q_1}).
$$ 
We have that 
$$
\Omega_{\mbox{\scriptsize{\rm Ts}}}=(e^{q_n}\d q_n+e^{-q_1}\d q_1)\wedge \sum_{i=1}^n \d p_i + 
e^{q_n-q_1}\d q_1\wedge\d q_n
$$ 
and that 
the deformed PN structure $({\RR}^{2n},\pi,N_{\mbox{\scriptsize{\rm Ts}}})$, where $N_{\mbox{\scriptsize{\rm Ts}}}=
N
+\pi^\sharp\Omega_{\mbox{\scriptsize{\rm Ts}}}^\flat$,
coincides with the one found by Tsiganov in \cite{Tsiganov2}, where it is used to obtain the integrals of motion of the 
closed Toda lattice 
by means of generalized Lenard-Magri recursion relations. 

In \cite{Tsiganov2} it is also noticed that the (canonical) transformation
\begin{equation}
\label{tsig-canon}
Q_j=q_j\ \forall j=1,\dots,n,\quad
P_j=p_j\ \forall j=2,\dots,n-1,\quad
P_1=p_1-e^{-q_1},\quad
P_n=p_n-e^{q_n},
\end{equation}
sends the Das-Okubo PN structure into $({\RR}^{2n},\pi,N_{\mbox{\scriptsize{\rm Ts}}})$. Thus the traces of the powers of 
$N_{\mbox{\scriptsize{\rm Ts}}}$ coincide, via the diffeomorphism \eqref{tsig-canon}, with the traces of the powers of 
$N$, i.e., with the integrals of motion of the open Toda lattice.

In the 3-particle case, we have that 
\begin{equation*}
\label{N-Toda-open}
N
=\left(\begin {array}{ccc|ccc} 
p_{{1}}&0&0&0&1&1\\ 
0&p_{{2}}&0&-1&0&1\\ 
0&0&p_{{3}}&-1&-1&0\\ 
\hline
0&-{e^{q_{{1}}-q_{{2}}}}&0&p_{{1}}&0&0\\ 
{e^{q_{{1}}-q_{{2}}}}&0&-{e^{q_{{2}}-q_{{3}}}}&0&p_{{2}}&0\\ 
0&{e^{q_{{2}}-q_{{3}}}}&0&0&0&p_{{3}}
\end {array}\right)
\end{equation*}
and
\begin{equation*}
\label{N-Tsiganov}
N_{\mbox{\scriptsize{\rm Ts}}}=
\left(\begin {array}{ccc|ccc} 
p_{{1}}+e^{-q_1}&0&e^{q_3}&0&1&1\\ 
e^{-q_1}&p_{{2}}&e^{q_3}&-1&0&1\\ 
e^{-q_1}&0&p_{{3}}+e^{q_3}&-1&-1&0\\ 
\hline
0&-{e^{q_{{1}}-q_{{2}}}}&e^{q_3-q_1}&p_{{1}}+e^{-q_1}&e^{-q_1}&e^{-q_1}\\ 
{e^{q_{{1}}-q_{{2}}}}&0&-{e^{q_{{2}}-q_{{3}}}}&0&p_{{2}}&0\\ 
-e^{q_3-q_1}&{e^{q_{{2}}-q_{{3}}}}&0&e^{q_3}&e^{q_3}&p_{{3}}+e^{q_3}
\end {array}\right).
\end{equation*}
\end{example}

\begin{remark}
The PN manifold obtained in \cite{Tsiganov2}, denoted above as $({\RR}^{2n},\pi,N_{\mbox{\scriptsize{\rm Ts}}})$, is a particular case of a general construction presented in \cite{Tsiganov1} and based on the theory of the Classical Yang-Baxter equation, more precisely on the Sklyanin brackets associated to its solutions.
\end{remark}


\begin{example}
\label{exa:ourToda}
We start again with the Das-Okubo PN structure of the open Toda lattice. If $f=e^{q_n}$, $g=e^{-q_1}$, and $I=\sum_{i=1}^n p_i$,
then one can check that Lemma \ref{lem:f-g-I} can be applied, with $\alpha=\d f$, $\beta=\d g$, and $\gamma=\d I$. 
As in Example \ref{exa:tsiganov}, we have that $\delta=\d h$, where 
$$
h=g+\frac12\{g,f\}=g=e^{-q_1}.
$$ 
So $\Omega=\alpha\wedge\beta=\d f\wedge\d g=e^{q_n-q_1}\d q_1\wedge \d q_n$, to be called $\Omega_1$ in the sequel. Hence
the deformed PqN structure is $({\RR}^{2n},\pi,N_-,\phi_-)$, where $N_-=
N
+\pi^\sharp\Omega_1^\flat$ and the 3-form $\phi_-$ turns out to be, see \eqref{N-phi-hat}, \eqref{eq:Maurer-Cartan},
and \eqref{eq:epsilon},
\begin{equation}
\label{phi-}
\phi_-=\d_N\Omega_1+\frac12[\Omega_1,\Omega_1]_\pi=\alpha\wedge\epsilon=\d f\wedge\left(2\d I\wedge\d g\right)
=2 e^{q_n-q_1}\d I\wedge\d q_n\wedge\d q_1.
\end{equation}
For more details on this PqN structure and its relations with the 
closed Toda lattice, see \cite{CFMP2025}. Notice that the 2-form $\Omega$ appearing in \cite{CFMP2025} is minus the 
$\Omega_1$ appearing here. 
In the 3-particle case, we have that
\begin{equation}
\label{N-}
N_-=
N
+\pi^\sharp\Omega_1^\flat=
\left(\begin {array}{ccc|ccc} 
p_{{1}}&0&0&0&1&1\\ 
0&p_{{2}}&0&-1&0&1\\ 
0&0&p_{{3}}&-1&-1&0\\ 
\hline
0&-{e^{q_{{1}}-q_{{2}}}}&e^{q_3-q_1}&p_{{1}}&0&0\\ 
{e^{q_{{1}}-q_{{2}}}}&0&-{e^{q_{{2}}-q_{{3}}}}&0&p_{{2}}&0\\ 
-e^{q_3-q_1}&{e^{q_{{2}}-q_{{3}}}}&0&0&0&p_{{3}}
\end {array}\right).
\end{equation}
If we take $f=-e^{q_n}$, then we obtain the PqN manifold called $({\RR}^{2n},\pi,N_+,\phi_+)$ in \cite{CFMP2025}.
\end{example}

\begin{rem}
\label{rem:N_2}
\begin{enumerate}
\item[]
\item 
As seen in Example \ref{exa:tsiganov}, one can obtain the Tsiganov PN structure from the Das-Okubo one applying the canonical transformation \eqref{tsig-canon}. It is an interesting, completely open problem, to find a characterization of the closed 2-forms inducing a deformation which can be described by a canonical transformation of the underlying Poisson manifold.
\item One can obtain the Tsiganov PN structure from the PqN structure $({\RR}^{2n},\pi,N_-,\phi_-)$ by means of the 2-form 
$\Omega_2=\Omega_{\mbox{\scriptsize{\rm Ts}}}-\Omega_1$, 
where $\Omega_{\mbox{\scriptsize{\rm Ts}}}$ and $\Omega_1$ are given in Examples \ref{exa:tsiganov} and \ref{exa:ourToda}, respectively.
It turns out that 
$$
\Omega_2=(e^{q_n}\d q_n+e^{-q_1}\d q_1)\wedge \sum_{i=1}^n \d p_i =\d f\wedge\d g,
$$
with $f$ and $g$ as in Example \ref{exa:tsiganov}. Moreover, if one deforms the Das-Okubo PN structure by means of $\Omega_2$, then one obtains another PqN structure $({\RR}^{2n},\pi,N_2,\phi_2)$. One can check that $\phi_2=-\phi_-$
and that, in the 3-particle case, 
\begin{equation}
\label{pi-omega}
N_2=
N
+\pi^\sharp\Omega_2^\flat=
\left(\begin {array}{ccc|ccc} 
p_{{1}}+e^{-q_1}&0&e^{q_3}&0&1&1\\ 
e^{-q_1}&p_{{2}}&e^{q_3}&-1&0&1\\ 
e^{-q_1}&0&p_{{3}}+e^{q_3}&-1&-1&0\\ 
\hline
0&-{e^{q_{{1}}-q_{{2}}}}&0&p_{{1}}+e^{-q_1}&e^{-q_1}&e^{-q_1}\\ 
{e^{q_{{1}}-q_{{2}}}}&0&-{e^{q_{{2}}-q_{{3}}}}&0&p_{{2}}&0\\ 
0&{e^{q_{{2}}-q_{{3}}}}&0&e^{q_3}&e^{q_3}&p_{{3}}+e^{q_3}
\end {array}\right).
\end{equation}
We know that $\Omega_{\mbox{\scriptsize{\rm Ts}}}$ gives rise to a PN structure, so it must be
$$
\d_{N
}\Omega_{\mbox{\scriptsize{\rm Ts}}}+\frac12[\Omega_{\mbox{\scriptsize{\rm Ts}}},\Omega_{\mbox{\scriptsize{\rm Ts}}}]_\pi=0.
$$
Since $\Omega_{\mbox{\scriptsize{\rm Ts}}}=\Omega_1+\Omega_2$, this fact can also be seen as a consequence of $\phi_2=-\phi_-$ and 
$[\Omega_1,\Omega_2]_\pi=0$ (see Remark 4 in \cite{DMP2024}). Figure  \ref{fig:defDO} summarizes the deformation processes. Notice that a vertical downward arrow labelled $\Omega_2-\Omega_1$ could also be added.
\end{enumerate}
\end{rem}


%

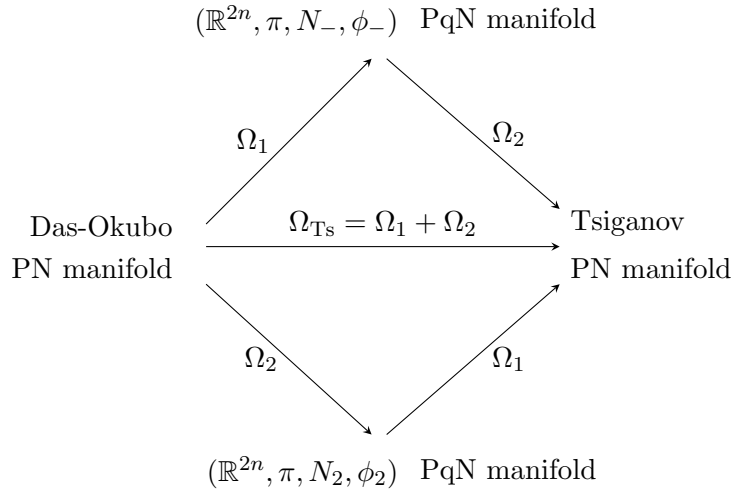
\begin{figure}[h]
\centering
\begin{minipage}{0.7\textwidth}
\begin{center}
\begin{tikzpicture}[scale=0.1]
\coordinate (D) at (-30,0);
\coordinate (T) at (20,0);
\coordinate (O) at (0,30);
\coordinate (N) at (0,-30);
\draw (D) node[above left]{\small Das-Okubo} node[below left]{\small PN manifold};
\draw (T) node[above right]{\small Tsiganov} node[below right]{\small PN manifold};
\draw (O) node[left]{\small $({\RR}^{2n},\pi,N_-,\phi_-)$} node[right]{\small PqN manifold};
\draw (N) node[left]{\small $({\RR}^{2n},\pi,N_2,\phi_2)$} node[right]{\small PqN manifold};
\draw[->,-stealth] (D)+(3,0) -- +(50,0) node[midway,above]{\small $\Omega_{\mbox{\scriptsize{\rm Ts}}}=\Omega_1+\Omega_2$};
\draw[->,-stealth] (D)+(3,3) -- +(25,25) node[midway,left]{\small $\Omega_1\;$};
\draw[->,-stealth] (O)+(-3,-5) -- +(20,-25) node[midway,right]{\small $\;\Omega_2$};
\draw[->,-stealth] (N)+(-3,5) -- +(20,25) node[midway,right]{\small $\;\Omega_1$};
\draw[->,-stealth] (D)+(3,-5) -- +(25,-25) node[midway,left]{\small $\quad\Omega_2$};
\end{tikzpicture}
\end{center}
\end{minipage}
\caption{
Deformations from the Das-Okubo PN structure.}\label{fig:defDO}
\end{figure}

\bigskip


%



\section{An involutivity theorem}
\label{sec:inv_thm}

In this section we present a simple sufficient condition for a PqN manifold to be involutive.
We start with

\begin{lem}
\label{lem:abg_involutivity}
Let $(\M,\pi,N,\phi)$ be a PqN manifold such that $\phi=\alpha\wedge\beta\wedge\gamma$, where $\alpha$, $\beta$, and $\gamma$ are 1-forms. 
If 
$$
H_k=\frac{1}{2k}\Tr(N^k),\qquad X_k=\pi^{\sharp}\d H_k,\qquad Y_k=N^{k-1}X_1-X_k,
$$
then 
\begin{equation}
\label{eq:alpha_Yk}
\langle\alpha,Y_k\rangle=0\qquad \text{for all }k\geq 1.
\end{equation}
\end{lem}
\begin{proof}
First of all, we show that 
\begin{equation}
\label{eq:T-fact}
T_N=\cycsum\pi^\sharp\alpha\otimes\beta\wedge\gamma.
\end{equation} 
Indeed, for any pair of vector fields $X$, $Y$,
\begin{equation*}
\begin{aligned}
\label{eq:TN-fact}
T_N(X,Y)
&=\pi^\sharp\left(i_{X\wedge Y}\phi\right)
=\pi^\sharp\left(i_{X\wedge Y}(\alpha\wedge\beta\wedge\gamma)\right)
=\pi^\sharp\left(i_Y i_X(\alpha\wedge\beta\wedge\gamma)\right)\\
&=\pi^\sharp\left(i_Y \left(\cycsum\langle\alpha,X\rangle\beta\wedge\gamma\right)\right)
=\pi^\sharp\left(\cycsum\langle\alpha,X\rangle\left(\brke{\beta}{Y}\gamma-\brke{\gamma}{Y}\beta\right)\right)\\
&=\cycsum\langle\alpha,X\rangle\left(\brke{\beta}{Y}\pi^\sharp\gamma-\brke{\gamma}{Y}\pi^\sharp\beta\right)
=\cycsum\left(\brke{\beta}{X}\brke{\gamma}{Y}-\brke{\beta}{Y}\brke{\gamma}{X}\right)\pi^\sharp\alpha\\
&=\cycsum\left(\beta\wedge\gamma\right)(X,Y)\pi^\sharp\alpha.
\end{aligned}
\end{equation*} 
It 
follows from \eqref{eq:T-fact} that
\begin{equation}
\label{eq:iT-fact}
i_X T_N=\cycsum\pi^\sharp\alpha\otimes\left(\brke{\beta}{X}\gamma-\brke{\gamma}{X}\beta\right).
\end{equation} 
Now considering the 1-forms \eqref{eq:phi_k} and 
using \eqref{eq:iT-fact}, we have that
\begin{equation*}
\label{varphi-bis}
\begin{aligned}
\brke{\phi_l}{X}
&= \frac12\Tr\left(N^{l}
\cycsum\pi^\sharp\alpha\otimes\left(\brke{\beta}{X}\gamma-\brke{\gamma}{X}\beta\right)\right)\\
&= \frac12\cycsum\left(
\brke{\beta}{X}\Tr\left(
N^{l}\pi^\sharp\alpha\otimes\gamma
\right)-
\brke{\gamma}{X}\Tr\left(
N^{l}\pi^\sharp\alpha\otimes\beta
\right)\right)\\
&= \frac12\cycsum\left(
\brke{\beta}{X}\brke{\gamma}{N^{l}\pi^\sharp\alpha}-\brke{\gamma}{X}\brke{\beta}{N^{l}\pi^\sharp\alpha}
\right)\\
&= \cycsum\brke{\alpha}{X}\brke{\beta}{N^{l}\pi^\sharp\gamma},
\end{aligned}
\end{equation*}
where in the last equality we used the skewsymmetry of $N^{l}\pi^\sharp$. Therefore, 
\begin{equation}
\label{varphi-ter}
\phi_l=\cycsum\brke{\beta}{N^{l}\pi^\sharp\gamma}\alpha.
\end{equation}
%
Finally, we recall (see \cite{CFMP2025}, before equation (56)) that 
\begin{equation}\label{eq:Yk}
Y_k=\pi^{\sharp}\sum_{l=0}^{k-2}\left(N^{*}\right)^{k-l-2}\phi_l.
\end{equation} 
Thus 
we have that
\begin{equation*}
\label{eq:alpha_Yk_b}
\begin{aligned}
&\langle\alpha,Y_k\rangle\\
&=\sum_{l=0}^{k-2}\langle\alpha,\pi^{\sharp}\left(N^{*}\right)^{k-l-2}\phi_l\rangle
=-\sum_{l=0}^{k-2}\langle\phi_l,N^{k-l-2}\pi^{\sharp}\alpha\rangle\\
&=-\sum_{l=0}^{k-2}\left(
\brke{\alpha}{N^{k-l-2}\pi^{\sharp}\alpha}\brke{\beta}{N^{l}\pi^\sharp\gamma}+
\brke{\beta}{N^{k-l-2}\pi^{\sharp}\alpha}\brke{\gamma}{N^{l}\pi^\sharp\alpha}+
\brke{\gamma}{N^{k-l-2}\pi^{\sharp}\alpha}\brke{\alpha}{N^{l}\pi^\sharp\beta}
\right)\\
&=-\sum_{l=0}^{k-2}\brke{\beta}{N^{k-l-2}\pi^{\sharp}\alpha}\brke{\gamma}{N^{l}\pi^\sharp\alpha}+
\sum_{l=0}^{k-2}\brke{\gamma}{N^{k-l-2}\pi^{\sharp}\alpha}\brke{\beta}{N^{l}\pi^\sharp\alpha}=0,
\end{aligned}
\end{equation*}
where in the last but one equality we used again the skewsymmetry of $N^{l}\pi^\sharp$.
\end{proof}

\begin{cor}
\label{cor:abg_Yk}
If $(\M,\pi,N,\phi)$ is a PqN manifold fulfilling the hypotheses of Lemma \ref{lem:abg_involutivity}, then, for all $k\geq 1$,
\begin{equation}
\label{eq:abg_Yk}
\langle\alpha,Y_k\rangle=\langle\beta,Y_k\rangle=\langle\gamma,Y_k\rangle=0,
\end{equation}
so that $i_{Y_k}\phi=0$. Moreover, it follows from \eqref{varphi-ter} that $\langle\phi_l,Y_k\rangle=0$, implying that
\begin{equation}
\label{eq:phil_Yk}
\langle\phi_l,X_k\rangle=\langle\phi_l,N^{k-1}X_1\rangle.
\end{equation}
\end{cor}

We are now ready for the main result of this paper.

\begin{theorem}
\label{thm:inv}
Let $(\M,\pi,N,\phi)$ be a PqN manifold such that $\phi=\alpha \wedge \beta\wedge \gamma$, where $\alpha=\d H_1$, 
$H_k=\frac{1}{2k}\Tr(N^k)$, and $\beta$, $\gamma$ are 1-forms. Then $(\M,\pi,N,\phi)$ is involutive, i.e., 
$$
\{H_l,H_m\}=0\qquad\mbox{for all $l,m\ge 1$.}
$$
\end{theorem}
\begin{proof}
We could show that the hypotheses of Theorem 4 in \cite{CFMP2025} are satisfied, but we think it is worthwhile to give a self-consistent proof.

We recall from \cite{CFMP2025}  (equation (16)) that
\begin{equation}
\label{recadd2}
\parpo{H_l}{H_m}
=\parpo{H_{l-1}}{H_{m+1}}-\brke{ \phi_{m-1}}{X_{l-1}}-\brke{\phi_{l-2}}{X_m}.
\end{equation}
Using \eqref{eq:phil_Yk}, we obtain
\begin{equation}
\label{recadd3}
\parpo{H_l}{H_m}
=\parpo{H_{l-1}}{H_{m+1}}-\brke{ \phi_{m-1}}{N^{l-2}X_1}-\brke{\phi_{l-2}}{N^{m-1}X_1}.
\end{equation}
From \eqref{varphi-ter}, and taking into account that $X_1=\pi^\sharp\d H_1=\pi^\sharp\alpha$, we have that
\begin{equation}
\begin{aligned}
&\brke{ \phi_{m-1}}{N^{l-2}X_1}\\
&=\brke{\beta}{N^{m-1}\pi^\sharp\gamma}\cancel{\brke{\alpha}{N^{l-2}X_1}}
+\brke{\gamma}{N^{m-1}\pi^\sharp\alpha}\brke{\beta}{N^{l-2}X_1}
+\brke{\alpha}{N^{m-1}\pi^\sharp\beta}\brke{\gamma}{N^{l-2}X_1}\\
&=\brke{\gamma}{N^{m-1}X_1}\brke{\beta}{N^{l-2}X_1}
-\brke{\beta}{N^{m-1}X_1}\brke{\gamma}{N^{l-2}X_1}.
\end{aligned}
\end{equation} 
The sign of the last row changes if $m-1$ and $l-2$ are exchanged. Therefore \eqref{recadd3} boils down to $\parpo{H_l}{H_m}
=\parpo{H_{l-1}}{H_{m+1}}$, which is well know to entail the involutivity of the functions $H_k$.
\end{proof}



We now present some examples of applications of Theorem \ref{thm:inv}. In all of these, the Poisson tensor $\pi$ is the canonical one.

\begin{example}
\label{exa:inv1}
Consider the PqN manifold $({\RR}^{2n},\pi,N_-,\phi_-)$ described in Example \ref{exa:ourToda}. Since 
$$
\phi_-=2\d I\wedge\d g\wedge\d f
$$
and $\d I=\d H^-_1$, where $H^-_k=\frac{1}{2k}\Tr(N_-^k)$, we can apply Theorem \ref{thm:inv} to prove that the PqN manifold is involutive. 
In this case 
\begin{equation}
\label{energy-closed-Toda}
H^-_2=\frac12\sum_{i=1}^{n}p_i^2+\sum_{i=1}^{n-1}{\rm e}^{q_i-q_{i+1}}- {\rm e}^{q_n-q_1}
\end{equation}
is the energy of a closed Toda lattice, where the interaction between the first and the last particle is repulsive. 
To obtain the energy of the classical closed Toda lattice, one has to consider suitable generalized recursion relations  
or, as an alternative, the traces of the powers of the (1,1) tensor field of the PqN structure $({\RR}^{2n},\pi,N_+,\phi_+)$, see \cite{CFMP2025} for more details.

\end{example}

\begin{example}
\label{exa:inv2}
We can apply Theorem \ref{thm:inv} also to the new PqN manifold $({\RR}^{2n},\pi,N_2,\phi_2)$ found in Remark \ref{rem:N_2}.
Indeed, using \eqref{phi-} we have that 
$$
\phi_2=-\phi_-=-2 e^{q_n-q_1}\d I \wedge\d q_n\wedge\d q_1=2 \d I\wedge\d e^{q_n}\wedge\d e^{-q_1},\qquad
\mbox{where $I=\sum_{i=1}^n p_i$.}
$$
Since $H_1=I+e^{-q_1}+e^{q_n}$, where $H_k=\frac{1}{2k}\Tr(N_2^k)$, we have that 
$\phi_2=2 \d H_1\wedge\d e^{q_n}\wedge\d e^{-q_1}$ and so 
the PqN manifold is involutive. We plan to study in the future the integrable system associated to this PqN structure. 
\end{example}


\begin{example}
\label{exa:inv3}
We consider now the PqN structures $({\RR}^{2n},\pi,\widehat{N},\widehat{\phi})$ of the (closed) Toda lattices associated to the Lie algebras $C_n^{(1)}$ and $A_{2n}^{(2)}$. These structures were found in \cite{CFMP2025}, where it was shown that $\widehat{\phi}=-2\d \widehat{H}_1\wedge\Omega$, with $\widehat{H}_k=\frac{1}{2k}\Tr(\widehat{N}^k)$ and 
$\Omega=\d(e^{-2q_1})\wedge\d p_1$. Thus Theorem \ref{thm:inv} immediately implies that these PqN manifolds are involutive, without the need of the lengthy computations in Appendix A of \cite{CFMP2025}.
\end{example}




In the following example, we construct a new involutive PqN structure by deforming the PN structure of the $D_n$-Toda system using the same closed 2-form as in Example \ref{exa:inv3}.

\begin{example}\label{exa:inv4}
We consider the PN structure $({\RR}^{2n},\pi,N)$ of the 
Toda lattice of type $D_n$, i.e., the Hamiltonian system whose Hamiltonian is $H_{D_n}=\frac{1}{2}\sum_{i=1}^np_i^2+\sum_{i=1}^{n-1}e^{q_i-q_{i+1}}+e^{q_{n-1}+q_n}$.  The Nijenhuis tensor $N$ is $(\pi')^\sharp(\pi^\sharp)^{-1}$, where the Poisson bracket associated to the Poisson bivector $\pi'$ 
can be found in \cite{DK}, equation (39).  By means of the closed 2-form $\Omega=\d(e^{-2q_1})\wedge \d p_1$,  
we deform $({\RR}^{2n},\pi,
N)$ to obtain the PqN structure $({\RR}^{2n},\pi,\widehat{N},\widehat{\phi})$. The $3$-form $\widehat{\phi}$ turns out to be 
\begin{equation*}
\begin{split}
\widehat{\phi}&=\d_N\Omega+\frac{1}{2}\left[\Omega,\Omega\right]_{\pi}=\d_N\Omega\\
&=\left(-4e^{-2q_1}\d q_1\wedge N^*\d q_1+2e^{-2q_1}\d(N^*\d q_1)\right)\wedge \d p_1-2e^{-2q_1}\d q_1\wedge \d(N^*\d p_1)\\
&=-8e^{-2q_1}\left(\sum_{i=2}^{n-1}\left(e^{q_{i-1}-q_i}-e^{q_i-q_{i+1}}\right)\d q_1\wedge \d q_i\wedge \d p_1-e^{q_{n-1}+q_n}\d q_1\wedge \d q_{n-1}\wedge \d p_1\right.\\
&\left.+\left(e^{q_{n-1}-q_n}-e^{q_{n-1}+q_n}\right)\d q_1\wedge \d q_n\wedge \d p_1+\sum_{i=2}^np_i\d q_1\wedge \d p_1\wedge \d p_i\right),
\end{split}
\end{equation*}
since
\begin{equation*}
\begin{split}
N^*\d q_1&=-(p_1^2+2e^{q_1-q_2})\d q_1+(e^{q_1-q_2}-2e^{q_2-q_3})\d q_2+2\sum_{i=3}^{n-1}(e^{q_{i-1}-q_i}-e^{q_i-q_{i+1}})\d q_i\\
&-2e^{q_{n-1}+q_n}dq_{n-1}+2(e^{q_{n-1}-q_n}-e^{q_{n-1}+q_n})\d q_n-2\sum_{i=1}p_i\d p_i\\
N^*\d p_1&=e^{q_1-q_2}(p_1+p_2)\d q_2-(p_1^2+2e^{q_1-q_2})dp_1-e^{q_1-q_2}\d p_2.
\end{split}
\end{equation*}
On the other hand, an easy computation shows that
\[
\widehat{H}_1=\frac{1}{2}\Tr(\widehat{N})=\sum_{i=1}^n\lbrace p_i,q_i\rbrace'-2e^{-2q_1}=-2H_{D_n}-2e^{-2q_1}
\]
and $\widehat{\phi}=-2\d\widehat{H}_1\wedge \Omega$.
Thus, Theorem \ref{thm:inv}  implies that the PqN structure $({\RR}^{2n},\pi,\widehat{N},\widehat{\phi})$ is involutive.
\end{example}

For the reader's convenience, we summarize in Table \ref{table1} the PN and PqN deformations discussed above.

\begin{table}[hbt!]
\begin{center}
\begin{tabular}{|c|c|c|}
\hline
{\small Structure}  & {\small 2-form $\Omega$} & {\small Deformed structure} \\[2mm]
\hline\hline
{\small Das-Okubo} & {\small $\Omega_{\mbox{\scriptsize{\rm Ts}}}$} & {\small Tsiganov PN structure}\\
{\small PN structure} & {\small  (Example \ref{exa:tsiganov})} &   {\small  $(\pi,N_{\mbox{\scriptsize{\rm Ts}}})$} \\  [2mm]
\hline
{\small Das-Okubo} & {\small $\Omega_1=\d(e^{q_n})\wedge \d(e^{-q_1})$}  & {\small   involutive PqN}\\ 
{\small PN structure}  &  {\small (Example \ref{exa:ourToda})} & {\small $(\pi,N_-,\phi_-)$}  \\[2mm]
\hline 
{\small Das-Okubo} & {\small $-\Omega_1$} &  {\small  involutive PqN $(\pi,N_+,\phi_+)$}\\ 
{\small PN structure}  & {\small  (Example \ref{exa:ourToda})} & {\small for $A_n^{(1)}$-Toda}  \\[2mm]
\hline 
{\small involutive PqN} & {\small $\Omega_2=\Omega_{\mbox{\scriptsize{\rm Ts}}}-\Omega_1$} & {\small involutive PqN}\\
{\small $(\pi,N_-,\phi_-)$} & {\small (Remark \ref{rem:N_2})} & {\small $(\pi,N_2,\phi_2)$} \\[2mm]
\hline
{\small PN for} & {\small $\Omega=\d(e^{-2q_1})\wedge \d p_1$} & {\small involutive PqN for}\\
{\small $B_n$-Toda \cite{NDC-D}} & {\small (Example \ref{exa:inv3})} & {\small $A_{2n}^{(2)}$-Toda} \\[2mm]
\hline
{\small PN for} & {\small $\Omega=\d(e^{-2q_1})\wedge \d p_1$} & {\small involutive PqN for}\\
{\small $C_n$-Toda \cite{NDC-D}} & {\small (Example \ref{exa:inv3})} & {\small $C_{n}^{(1)}$-Toda}\\[2mm]
\hline
{\small PN for} & {\small $\Omega=\d(e^{-2q_1})\wedge \d p_1$} & {\small involutive PqN for system with Hamiltonian}\\
{\small $D_n$-Toda \cite{DK}} & {\small (Example \ref{exa:inv4})} & {\small $\frac{1}{2}\sum_{i=1}^{n}p_i^2+\sum_{i=1}^{n-1}e^{q_i-q_{i+1}}+e^{q_{n-1}+q_n}+e^{-2q_1}$} \\ [2mm]
\hline
\end{tabular}
\end{center}
\caption{List of deformations}\label{table1}
\end{table}



We finish our paper with the following result, which is a straightforward consequence of Proposition \ref{prop:ab_def} and Theorem \ref{thm:inv}. 


\begin{cor}\label{cor:inv_def}
Let  $(\M,\pi,N)$ be a PN manifold. Suppose that $\alpha$, $\beta$, $\delta$ are closed 1-forms and 
$\Omega$ a closed 2-form such that the hypotheses of Proposition \ref{prop:ab_def}, see \eqref{recurrence}  and  \eqref{eq:deltaOmega}, are fulfilled. Let $(\M,\pi,\widehat N,\widehat \phi)$ be the deformed PqN manifold, defined by
	$$
	\widehat N=N+\pi^\sharp\Omega^\flat,\qquad
	\widehat \phi=
	\frac12\alpha\wedge \left[[\beta,\alpha]_\pi,\alpha\right]_\pi\wedge\delta.
	$$
	If $\alpha=\d{\widehat H}_1$ and ${\widehat H}_k=\frac{1}{2k}\Tr({\widehat N}^k)$, then 
	$$
	\{{\widehat H}_l,{\widehat H}_m\}=0\qquad\mbox{for all $l,m\ge 1$.}
	$$
\end{cor}


The next and last example aims to present an application of the previous corollary. More precisely, we will show that using equations 
\eqref{eq:deltaOmega} 
one can produce a PN deformation of a PN structure relevant to the theory of separation of variables, see Section 6 in \cite{Crampin-Sarlet-Thompson}.

\begin{example}\label{exa:inv5}
Let $Q$ be $\mathbb R^2$ with coordinates $(q_1,q_2)$. Let $L$ be the (1,1) tensor on $Q$ whose matrix expression is 
\[
L=\begin{pmatrix} q_1^2 & q_1q_2\\
	q_1q_2 & q_2^2
	\end{pmatrix}.
\]	
An easy computation shows that its cotangent lift to $T^\ast Q=\mathbb R^4$  has the 
form
\[
N=\begin{pmatrix}
	q_1^2 & q_1q_2 & 0 & 0\\
	q_1q_2 & q_2^2 & 0 & 0\\
	0 & q_1p_2-q_2p_1 & q^2_1 & q_1q_2\\
	q_2p_1-q_1p_2 & 0 & q_1q_2 & q_2^2
\end{pmatrix},
\]
which is torsion free because $L$ is --- see, for example, \cite{Crampin-Sarlet-Thompson,Ibort-Magri-Marmo}. In fact, as one can easily check, 
\[
\left[L\frac{\partial}{\partial q_1},L\frac{\partial}{\partial q_2}\right]=0\quad\text{and}\quad L\left(\left[L\frac{\partial}{\partial q_1},\frac{\partial}{\partial q_2}\right]+\left[\frac{\partial}{\partial q_1},L\frac{\partial}{\partial q_2}\right]\right)=0.
\]
Now let
\begin{equation}
h=\frac{1}{2}(p_1^2+p_2^2)+\frac{1}{2}(q_1^{-2}+q_2^{-2}).\label{eq:hami}
\end{equation}
Its differential is 
\[
\d h=p_1\d p_1+p_2\d p_2-q_1^{-3}\d q_1-q_2^{-3}\d q_2,
\]
whose image via $N^\ast$ is 
\begin{eqnarray*}
N^\ast \d h&=&(-q_1^{-1}-q_1q_2^{-2}+q_2p_1p_2-q_1p_2^2)\d q_1+(-q_2q_1^{-2}-q_2^{-1}+q_1p_1p_2-q_2p_1^2)\d q_2\\
&+&(q_1^2p_1+q_1q_2p_2)\d p_1+(q_1q_2p_1+q_2^2p_2)\d p_2.
\end{eqnarray*}
A tedious but straightforward computation, which uses the formula above, shows that 
\begin{equation}
	\d_N\d h=-\d(N^\ast \d h)=\d h\wedge \d \Tr L.\label{eq:e1}
\end{equation}
Furthermore, since $N$ is a Nijenhuis tensor and $\Tr N=2\Tr L$,
\begin{equation}
	\d_N\d\Tr L=0.\label{eq:e2}
\end{equation} 
In other words, $\alpha=\d h$ and $\beta=\d\Tr L$ satisfy 
equations \eqref{recurrence}. Now one can observe that 
\[
[\beta,\alpha]_\pi=\frac{1}{2}\left[\d(q_1^2+q_2^2),\d(p_1^2+p_2^2)+\d(q_1^{-2}+q_2^{-2})\right]_\pi=\frac{1}{2}\d\{q_1^2+q_2^2,p_1^2+p_2^2\}=0,
\]
where $\pi$ is the canonical Poisson tensor on $\mathbb R^4$. The previous computation yields 
$\delta=\beta$ and $\Omega=\alpha\wedge\beta$, see formulas 
\eqref{eq:deltaOmega}. Written in coordinates,
\[
\Omega=-2(p_1q_1\d p_1\wedge \d q_1+p_2q_1\d p_2\wedge \d q_1+p_1q_2\d p_1\wedge \d q_2+p_2q_2\d p_2\wedge \d q_2),
\]
which entails
\begin{equation}
	{\widehat N}=N+\pi^\sharp\Omega^\flat=\begin{pmatrix}
		q_1^2+2p_1q_1 & q_1q_2+2p_1q_2 & 0 & 0\\
		q_1q_2+2p_2q_1 & q_2^2+2p_2q_2 & 0 & 0\\
		0 & q_1p_2-q_2p_1 & q^2_1+2p_1q_1 & q_1q_2+2p_1q_2\\
		q_2p_1-q_1p_2 & 0 & q_1q_2+2p_2q_1 & q_2^2+2p_2q_2
	\end{pmatrix}.
	\end{equation}
	Finally, note that $\widehat\phi=0$ since $[\beta,\alpha]_\pi=0$. This implies that $(\pi,\widehat N)$ is a PN deformation of the original PN structure $(\pi,N)$, see also Remark \ref{rem:PNtoPN}.
\end{example}

\begin{remark}
The initial PN structure of the previous example belongs to a class of Poisson-Nijenhuis structures which play an important role in the theory of separation of variables, see \cite{Ibort-Magri-Marmo, Crampin-Sarlet-Thompson, FP03} and reference therein. Equations \eqref{recurrence} are a crucial ingredient in the bi-Hamiltonian approach to the theory of separation of variables, see again \cite{Ibort-Magri-Marmo, Crampin-Sarlet-Thompson}. As we have already remarked, the example above illustrates a PN deformation of a PN structure described in Section 6 of \cite{Crampin-Sarlet-Thompson}. It would be interesting to find involutive PqN deformations of PN structures relevant in the theory of separation of variables.
\end{remark}

\par\medskip\noindent
{\bf Acknowledgments.} 
We thank Andrea Raimondo, Pavel Etingof, Semenov Tian-Shansky and Pol Vanhaecke for useful comments and interesting discussions. MP thanks the ICMC-USP, {\em Instituto de Ci\^encias Matem\'aticas e de Computa\c c\~ao\/} of the University of S\~ao Paulo, and the Department of Mathematics and its Applications of the University of Milano-Bicocca for their hospitality, and the University of Bergamo, for supporting his visit in 2025 to the ICMC-USP within the program {\em Outgoing Visiting Professors}. IM thanks the Department of Mathematics and its Applications of the University of Milano-Bicocca for its hospitality and the University of Bergamo for supporting his visit in 2026.
This project has received funding from the European Union's Horizon 2020 research and innovation programme under the 
Marie Sk{\l}odowska-Curie grant no 778010 {\em IPaDEGAN} as well as by the Italian PRIN 2022 (2022TEB52W) - PE1 - project {\em The charm of integrability: from nonlinear waves to random matrices}. All authors gratefully acknowledge the auspices of the GNFM Section of INdAM, under which part of this work was carried out, and the financial support of the project MMNLP (Mathematical Methods in Non Linear Physics) of the INFN. 


%
%





\thebibliography{99}

%
%
%
%
%
\bibitem{BursztynDrummond2019}  Bursztyn, H., Drummond, T.,
{\it Lie theory of multiplicative tensors}, Math.\ Ann.\ {\bf 375} (2019), 1489--1554.
%

\bibitem{CFMP2025} Chu\~no Vizarreta, E., Falqui, G., Mencattini, I., Pedroni, M.,
{\it Poisson quasi-Nijenhuis manifolds, closed Toda lattices, and generalized recursion relations}, 
Lett.\ Math.\ Phys.\ {\bf 115} (2025), 84 (22 pages).


\bibitem{Crampin-Sarlet-Thompson}  Crampin, M., Sarlet, W., Thompson, G., 
{\it Bi-differential calculi, bi-Hamiltonian systems and conformal Killing tensors},
J. Phys. A {\bf 33} (2000), 8755--8770.





\bibitem{DK} Damianou, P. A., Kouzaris, S. P., {\it Bogoyavlensky-Toda systems of type $D_N$},  J. Phys. A: Math. Gen. {\bf 36} (2003), 1385--1399.

\bibitem{DO} Das, A., Okubo, S., {\it A systematic study of the Toda lattice}, Ann. Physics {\bf 190} (1989), 215--232.

\bibitem{DMP2024} do Nascimento Luiz, M., Mencattini, I., Pedroni, M.,
{\it Quasi-Lie bialgebroids, Dirac structures, and deformations of Poisson quasi-Nijenhuis manifolds\/}, 
Bull.\ Braz.\ Math.\ Soc.\ (N.S.) {\bf 55} (2024), 18 pages.

\bibitem{DZ} Dufour, J.-P., Zung, N.T., {\it Poisson Structures and their Normal Forms}, Progress in Mathematics, Volume 242, Birkh\"auser Verlag, 2005.

%
%
\bibitem{FMOP2020} Falqui, G., Mencattini, I., Ortenzi, G., Pedroni, M.,
{\it Poisson Quasi-Nijenhuis Manifolds and the Toda System\/}, Math.\ Phys.\ Anal.\ Geom.\ {\bf 23} (2020), 17 pages.

\bibitem{FMP2023} Falqui, G., Mencattini, I., Pedroni, M.,
{\it Poisson quasi-Nijenhuis deformations of the canonical PN structure\/}, J.\ Geom.\ Phys.\  {\bf 186} (2023), 10 pages.

\bibitem{FP03} Falqui, G., Pedroni, M., {\it Separation of variables for bi-Hamiltonian systems\/}, 
Math.\ Phys.\ Anal.\ Geom.\ {\bf 6} (2003), 139--179.

\bibitem{FiorenzaManetti2012} Fiorenza, D., Manetti, M.,
{\it Formality of Koszul brackets and deformations of holomorphic Poisson manifolds},
Homology Homotopy Appl. {\bf 14} (2012), 63--75. 


\bibitem{Ibort-Magri-Marmo} Ibort, A., Magri, F., Marmo, G., {\it bi-Hamiltonian structures and St\"ackel separability}, J.\ Geom.\ Phys.\  {\bf 33} (2000), 210--218.

\bibitem{YKS96} Kosmann-Schwarzbach, Y., {\it The Lie Bialgebroid of a Poisson-Nijenhuis Manifold}, Lett. Math. Phys. {\bf 38} (1996), 421--428.

\bibitem{KM} Kosmann-Schwarzbach, Y., Magri, F., {\it Poisson-Nijenhuis structures}, Ann. Inst. Henri Poincar\'e {\bf 53} (1990), 35--81.






\bibitem{MagriMorosiRagnisco85} Magri, F., Morosi, C., Ragnisco, O.,
{\it Reduction techniques for infinite-dimensional Hamiltonian systems: some ideas and applications},
Comm. Math. Phys. {\bf 99} (1985), 115--140. 
%

\bibitem{NDC-D} Nunes da Costa, J.M., Damianou, P.A., {\it Toda systems and exponents of simple Lie groups}, 
Bull. Sci. Math. {\bf 125} (2001), 49--69.




\bibitem{RSTS} Reyman, A.G., Semenov-Tian-Shansky, M.A., 
\emph{Group-Theoretical Methods in the Theory of Finite-Dimensional Integrable Systems}. 
In: Dynamical Systems VII, Encyclopaedia of Mathematical Sciences, vol.\ 16 (Arnol’d, V.I., Novikov, S.P., eds.), 
Springer, Berlin, 1994.
%
\bibitem{SX} Sti\'enon, M., Xu, P., {\it Poisson Quasi-Nijenhuis Manifolds}, Commun. Math. Phys. {\bf 270} (2007), 709--725.
%

\bibitem{Tsiganov1} Tsiganov, A.V.,
{\it A family of the Poisson brackets compatible with the Sklyanin bracket},
J. Phys. A {\bf 40} (2007), 4803--4816.

\bibitem{Tsiganov2} Tsiganov, A.V.,
{\it On two different bi-Hamiltonian structures for the Toda lattice},
J. Phys. A {\bf 40} (2007), 6395--6406.
\end{document}